\documentclass[12pt,letterpaper]{article}
\usepackage[top=2cm, bottom=4.5cm, left=2.5cm, right=2.5cm]{geometry}
\usepackage[utf8]{inputenc}

\usepackage{amssymb, nicefrac, mathtools, amsmath, graphicx, dsfont, bbold, mathrsfs, amsthm, hyperref, biblatex, comment, quoting, framed}
\usepackage{xcolor}
\graphicspath{ {images/} }

\def\R{\mathbb{R}}

\def\Q{\mathbb{Q}}
\def\Z{\mathbb{Z}}
\def\P{\mathbb{P}}

\def\bigspace{\,\,\,\,\,\,\,\,}

\DeclareMathOperator*{\argmax}{arg\,max}
\DeclareMathOperator*{\argmin}{arg\,min}

\newtheorem{thm}{Theorem}
\newtheorem{example}[thm]{Example}
\newtheorem{lemma}[thm]{Lemma}

\newtheorem{corollary}[thm]{Corollary}
\newtheorem{definition}{Definition}
\newtheorem*{remark}{Remark}

\newcommand{\thesistitle}{Justifications for Generalizations of Approval Voting}
\newcommand{\thesisauthor}{Hari Sarang Nathan}
\newcommand{\thesisadvisor}{Professor Michael E. Orrison}
\newcommand{\thesisadvisorU}{Harvey Mudd College}
\newcommand{\graddate}{May 2023}

 \colorlet{shadecolor}{lightgray}
\newenvironment{shadedquotation}
 {\begin{shaded*}
  \quoting[leftmargin=0pt, vskip=0pt]
 }
 {\endquoting
 \end{shaded*}
}

\begin{document}

\begin{center}
{\large\textbf{\thesistitle}}
\vspace{.2in}

by

\vspace{.1in}
\thesisauthor \\~\\~\\
\centering

A dissertation submitted in partial fulfillment\\~\\~\\ 
of the requirements for the degree of\\~\\~\\ 
Master of Science (M.S.) in Mathematics\\~\\~\\ 
Courant Institute of Mathematical Sciences\\~\\~\\ 
New York University\\~\\~\\ 
\graddate\\~\\~\\~\\~\\~\\~\\~\\~\\~\\~\\~\\
\end{center}

\begin{tabular}{@{}p{4in}@{}}
\hrulefill \\
\thesisadvisor \\
\thesisadvisorU
\end{tabular}

\pagebreak

\tableofcontents

\pagebreak

\addcontentsline{toc}{section}{Acknowledgments}
\section*{Acknowledgements}
I am greatly indebted to my advisor, Professor Michael E. Orrison for helping guide me on this thesis as well as Professor Katharine Shultis and Professor Jessica Sorrells for their ideas and support throughout this work. I am also grateful to Professor Vlad Vicol, Professor Sylvain Cappell, and Professor Ben Blum-Smith for their encouragement and in helping me find my thesis advisor. Finally, I could never have finished this without the support of my mother, Ms. Pushpa Nathan.

\addcontentsline{toc}{section}{Abstract}
\section*{Abstract}
Approval voting is a common method of preference aggregation where voters vote by ``approving'' of a subset of candidates and the winner(s) are those who are approved of by the largest number of voters. In approval voting, the degree to which a vote impacts a candidate's score depends only on if that voter approved of the candidate or not, i.e., it is independent of which, or how many, other candidates they approved of. Recently, there has been interest in satisfaction approval voting and quadratic voting both of which include a trade-off between approving of more candidates and how much support each selected candidate gets. Approval voting, satisfaction approval voting, and quadratic voting, can all be viewed as voting where a vote is viewed as analogous to a vector with a different unit norm ($\mathcal{L}^{\infty}$, $\mathcal{L}^{1}$, and $\mathcal{L}^2$ respectively). This suggests a generalization where one can view a vote as analogous to a normalized unit vector under an arbitrary $\mathcal{L}^p$-norm. In this paper, we look at various general methods for justifying voting methods and investigate the degree to which these serve as justifications for  these generalizations of approval voting.

\section{Introduction}
Approval voting is a common method of preference aggregation where voters vote by ``approving'' of a subset of candidates and the winner(s) are those who are approved of by the most voters. In approval voting, the degree to which a vote impacts a candidate's score depends only on if that voter approved of the candidate or not, i.e., it is independent of which, or how many, other candidates they approved of. However, other work has investigated the advantages of alternatives where the contribution to a candidates depends on the number of candidates approved.

In particular, \cite{Brams2010-cr} introduces ``satisfaction approval voting'' where a voter contributes a total of a single point divided equally among all the candidates they approved of. In addition, there has been recent interest in ``quadratic voting'' (e.g., \cite{Posner2017-ny}, \cite{weyl-glen-2013}) where (in some forms) each voter has a fixed pool of points and the contribution to a candidate is the square root of the number of points spent on that candidate (e.g., \cite{Quarfoot2017-bd}).

Approval voting, satisfaction approval voting, and quadratic voting, can all be viewed as voting where a vote is viewed as a unit vector with a different norm ($\mathcal{L}^{\infty}$, $\mathcal{L}^{1}$, and $\mathcal{L}^2$ respectively). This suggests a generalization where one can view a vote as a unit vector under an arbitrary $\mathcal{L}^p$-norm. In this paper, we study such generalizations which we call $k,p$-approval voting. In particular, voters pick a subset of candidates to approve of and the vote is viewed as analogous to a $0, 1$-vector which is then normalized according to some predetermined ($\mathcal{L}^p$) norm. We look at various general methods for justifying voting methods and investigate the degree to which these serve as justifications for $k,p$-approval voting.

The rest of this paper proceeds as follows. In section 2, we lay out the basic definitions and notation. In section 3, we provide an axiomatic justification of $k,p$-approval voting. In section 4, we look at justifications of $k,p$-approval voting from the perspective of utility maximization and maximum likelihood estimation. In section 5, we look at $k,p$-approval voting from the perspective of distance rationalizability. In section 6, we conclude.

As a note, throughout this paper we use the vocabulary of ``democratic'' elections, i.e., voters, candidates, etc. However, as \cite{Lackner2022-nb} points out, there are a number of other uses for approval voting. They list the following which span numerous fields:\footnote{In the quote below, we omit the citations to various references. The references can be found in \cite{Lackner2022-nb}.}

\begin{shadedquotation}\begin{enumerate}
    \item finding group recommendations where the possible recommendations can be thought of as candidates and individual group members as voters,
    \item collaborative filtering where, for example, related movies are recommended based on large data collections,
    \item diversifying search results where users sending a search query can be interpreted as voters and the possible search results correspond to candidates,
    \item locating public facilities where the candidates are possible locations in which facilities can be built,
    \item the design of dynamic Q\&A platforms, where participants propose and upvote questions to be asked in a Q\&A session,
    \item selecting validators in consensus protocols (blockchain), with the users of the protocol corresponding to both voters and candidates, and
    \item genetic programming, a technique to solve global optimisation problems.
\end{enumerate}
\end{shadedquotation}

\section{Definitions and Notation}
Before we get to the specific definitions we need, we lay out some notation we will use throughout the paper. \begin{itemize}
    \item We use $\mathcal{P}(X)$ to refer to the power set of $X$, i.e., $\mathcal{P}(X) = \{ Y \subset X \}$.
    \item We use $\mathds{1}$ as the indicator function, i.e.,

    $$\mathds{1}_{g(x)} = \begin{cases}
        1 & g(x) \text{ is true} \\
        0 & \text{otherwise}.
    \end{cases}$$
    \item Given two sets, $X$ and $Y$, we use $X^Y$ to refer to the set of all functions from $Y$ to $X$.
    \item Given two sets, $X$ and $Y$, we use $X \setminus Y$ as the set difference i.e. $X \setminus Y = \{x \in X : x \notin Y\}$.
    \item Given a vector $v \in \R^n$ we use $\| v \|_{p}$ as the $\mathcal{L}^p$ norm of the vector, i.e., for $p \in [1, \infty)$

    $$\| v \|_{p} = \left( \sum_{i = 1}^{n}{|v_i|^p} \right)^{1/p}$$

    and, for $p = \infty$,

    $$\| v \|_{\infty} = \lim_{p \to \infty}{\| v \|_{p}} = \max_{i \in \{1, ..., n\}}{|v_i|}.$$
    \item In this paper, we use both $\argmin$ and $\argmax$ as per the standard definitions. Of note, however, is that we treat these as returning the set of all arguments that minimize/maximize the relevant quantity. For example, we have

    $$\argmin_{x \in [0, 10]}{\{x - \lfloor x \rfloor\}} = \{0, 1, ..., 9, 10\},$$

    i.e., \textit{all} the integers in $[0, 10]$.
    \item We write $A \sim B$ to mean ``$A$ is proportional to $B$'' i.e. there is some $c > 0$ such that $A = cB$.
\end{itemize}

We now turn to the abstract definition of a voting system. A voting system is a way to take a finite collection of ballots and aggregate them into an outcome. To specify all of these, we use the following ingredients (\cite{Zwicker2008-xe}).

\begin{definition}
\label{def:aavs_elements}
We define: \begin{itemize}
    \item $\mathcal{O}$ a finite set, the set of \textbf{outputs};
    \item $\mathcal{I}$ a finite set, the set of \textbf{inputs} or \textbf{ballots};
    \item $\pi:\mathcal{I} \to \Z_{\geq 0}$ a \textbf{profile}; and
    \item $\Pi = \left(\Z_{\geq 0}\right)^{\mathcal{I}}$, i.e., the set of all functions from $\mathcal{I}$ to $\Z_{\geq 0}$, the set of \textbf{all profiles};
    \item $\mathcal{F}:\Pi \to \mathcal{P}(\mathcal{O}) \setminus \{ \emptyset \}$ is a \textbf{voting rule}.
\end{itemize}
\end{definition}

\begin{definition}
\label{def:aavs} A tuple $(\mathcal{I}, \mathcal{O}, \mathcal{F})$ is an \textbf{abstract anonymous voting system} (or just a \textbf{voting system}).\footnote{The definition in \cite{Zwicker2008-xe} also includes the ``set of possible outcomes'' which is denoted as  $\mathcal{C}(\mathcal{O})$. However, we find this notation redundant, since $\mathcal{F}(\Pi)$ plays this role and, below, we use $\mathcal{C}$ for the set of candidates.} \end{definition}

The ``anonymous'' part in the definition above refers to the fact that the voting system doesn't care which voters cast which vote, only how many voters cast each vote. This comes from the fact that $\mathcal{F}$ takes profiles, and not lists of votes, as inputs.\footnote{An example of a system that is not anonymous is a dictatorship where multiple people vote but only the dictator's vote counts and every other vote is ignored.} Note that $\mathcal{F}$ may not return a single output. When $|\mathcal{F}(\pi)| > 1$, we interpret this to mean that the various outcomes in $\mathcal{F}(\pi)$ are ``tied'' (at least as far as $\mathcal{F}$ and $\pi$ are concerned).

The above is somewhat abstract; a few examples will help clarify.

\begin{example}
\label{ex:plurality_voting} Perhaps the best-known voting system is \textbf{plurality voting} where each voter votes for a single candidate and the candidate with the most votes wins. In this system: \begin{itemize}
    \item $\mathcal{O} = \{c_1, ..., c_n\}$ is the set of candidates;
    \item $\mathcal{I} = \mathcal{O}$ since each vote is just for a single candidate; and 
    \item $\mathcal{F}$ returns the candidates with the most votes. This can be written as
    $$\mathcal{F}(\pi) = \argmax_{c \in \mathcal{O}}{\{\pi(c)\}}.$$
\end{itemize}

For example, if there are three candidates, $\{c_1, c_2, c_3\}$, we can look at two profiles.

\begin{center}
    \begin{tabular}{c|c|c|c|c}
        Profile & $c_1$ & $c_2$ & $c_3$ & $\mathcal{F}(\pi)$ \\
        \hline\hline
        $\pi_1$ & 10 & 15 & 5 & $\{c_2\}$ \\
        $\pi_2$ & 20 & 15 & 20 & $\{c_1, c_3\}$
    \end{tabular}
\end{center}

Here $\mathcal{F}(\pi_1) = \{c_2\}$ (the clear winner) but $\mathcal{F}(\pi_2) = \{c_1, c_3\}$ since they are tied for winner. If this were being used in a situation where a single clear winner was needed, some tie-breaking methodology would be needed.
\end{example}

\begin{example}
\label{ex:approval_voting}
We now turn to an example of the type we will study in this paper. In \textbf{approval voting} each candidate selects at least one candidate to approve of and the candidate approved of by the largest number voters wins. In this system: \begin{itemize}
    \item $\mathcal{O} = \{c_1, ..., c_n\}$ is the set of candidates;
    \item $\mathcal{I} = \mathcal{P}(\mathcal{O}) \setminus \{ \emptyset \}$ since each voter must approve of at least one candidate but can approve of as many as they like; and 
    \item $\mathcal{F}$ is defined via a function $s:\mathcal{O} \times \Pi \to \R$ where
    $$s(c, \pi) = \sum_{b \in \mathcal{I}}{\pi(b) \cdot \mathds{1}_{c \in b}}$$
    and
    $$\mathcal{F}(\pi) = \argmax_{c \in \mathcal{O}}{\{s(c, \pi)\}}.$$
\end{itemize}
The function $s$ is an example of what we will call a ``score function'' (which we will define formally below). Here, the score of each candidate is number of voters who approved of them and $\mathcal{F}$ returns all the candidates with maximum score. 

Again, we use a three candidate example with $\{c_1, c_2, c_3\}$ and look at two examples.

\begin{center}
    \begin{tabular}{c|c|c|c|c|c|c|c}
        Profile & $\{c_1\}$ & $\{c_2\}$ & $\{c_3\}$ & $\{c_1, c_2\}$ & $\{c_1, c_3\}$ & $\{c_2, c_3\}$ & $\{c_1, c_2, c_3\}$ \\
        \hline
        $\pi_1$ & 2 & 0 & 0 & 2 & 0 & 3 & 0 \\
        $\pi_2$ & 2 & 0 & 0 & 0 & 0 & 3 & 0
    \end{tabular}
\end{center}

This leads to the following results:

\begin{center}
    \begin{tabular}{c|c|c|c|c}
        Profile & $s(c_1, \pi)$ & $s(c_2, \pi)$ & $s(c_3, \pi)$ & $\mathcal{F}(\pi)$ \\
        \hline
        $\pi_1$ & 4 & 5 & 3 & $\{c_2\}$ \\
        $\pi_2$ & 2 & 3 & 3 & $\{c_2, c_3\}$
    \end{tabular}
\end{center}
\end{example} 

\begin{example}
\label{ex:sat_approval_voting}
An alternative to approval voting is to ``normalize'' votes which approve of more than one candidate. In \textbf{satisfaction approval voting} (\cite{Brams2010-cr}) $\mathcal{O}$ and $\mathcal{I}$ are the same as in example \ref{ex:approval_voting} above but we modify $s$ such that

$$s(c, \pi) = \sum_{b \in \mathcal{I}}{\frac{\pi(b) \cdot \mathds{1}_{c \in b}}{| b |}}.$$

This can change the outcome of an election. For example, for the profiles in the example \ref{ex:approval_voting}:

\begin{center}
    \begin{tabular}{c|c|c|c|c}
        Profile & $s(c_1, \pi)$ & $s(c_2, \pi)$ & $s(c_3, \pi)$ & $\mathcal{F}(\pi)$ \\
        \hline
        $\pi_1$ & 3 & 2 & 1.5 & $\{c_1\}$ \\
        $\pi_2$ & 2 & 1.5 & 1.5 & $\{c_1\}$
    \end{tabular}
\end{center}
\end{example}

With these definitions and examples in hand, we now turn to defining the specific objects of our study. In particular, we will study generalizations of approval voting and satisfaction approval voting. We generalize these in two ways. First, we allow for a broader range of score functions. Second, we allow for the possibility of selecting a committee of multiple candidates, not just a single candidate.

\begin{definition}
\label{def:p_aavs}
For $k \in \Z_{> 0}$ and $p \in [1, \infty]$, a \textbf{$k,p$-approval voting system}, $(\mathcal{I}, \mathcal{O}, \mathcal{F}_{k,p})$, is an abstract anonymous voting system where: \begin{itemize}
    \item $k$ is the \textbf{committee size};
    \item there are $n$ \textbf{candidates}, i.e., $\mathcal{C} = \{c_1, ..., c_n\}$;
    \item the possible outcomes are committees of size $k$, i.e., $\mathcal{O} = \{K \subset \mathcal{C} : |K| = k\}$;
    \item the ballots are \textbf{approval ballots}, i.e., $\mathcal{I} = \mathcal{P}(\mathcal{C}) \setminus \{ \emptyset \}$;\footnote{Everything here can be equally well formulated if we allow $\mathcal{I} = \mathcal{P}(\mathcal{C})$, i.e., if we allow a voter to ``approve of nobody.'' We only disallow it here because we will be dividing by $|b|^{1/p}$ and we would need to create special cases to avoid diving by zero.}
    \item the \textbf{score function} is $s_p:\mathcal{C} \times \Pi \to \R$ defined by as follows.\footnote{The inspiration for these score functions comes from thinking of the ballots as $0-1$ vectors in $\R^{\mathcal{C}} \cong \R^n$. For each $b \in \mathcal{I}$, we can consider the vector $\vec{b}$ where $\vec{b}_i = 1$ if $c_i \in b$ and $b_i = 0$ otherwise. In this case, we can see that $\|\vec{b}\|_p = |b|^{1/p}$ and we are, in effect, normalizing each $\vec{b}$ by its $\mathcal{L}^p$ norm.}  For $p < \infty$,

    $$s_p(c, \pi) = \sum_{b \in \mathcal{I}}{\frac{\pi(b) \cdot \mathds{1}_{c \in b}}{| b |^{1/p}}}$$

    and, for $p = \infty$,

    $$s_{\infty}(c, \pi) = \lim_{p \to \infty}{s_p(c, \pi)} = \sum_{b \in \mathcal{I}}{\pi(b) \cdot \mathds{1}_{c \in b}};$$

    and
    \item $\mathcal{F}_{k,p}$ is defined as

    $$\mathcal{F}_{k,p}(\pi) = \argmax_{K \in\mathcal{O}}{\left\{ \sum_{c \in K}{s_p(c, \pi)} \right\}}.$$
\end{itemize}
\end{definition}

\begin{example}
Using the definition above, $1,\infty$-approval voting is approval voting, $1,1$-approval voting is satisfaction approval voting. In addition, $1,2$-approval voting is quadratic voting.
\end{example}

\begin{example}
Changing $p$ can dramatically change the outcome. We look at the following profile, $\pi$:

\begin{center}
\begin{tabular}{c|c|c|c|c|c}
    $\{c_1\}$ & $\{c_2\}$ & $\{c_3\}$ & $\{c_1, c_2\}$ & $\{c_1, c_3\}$ & $\{c_2, c_3\}$\\
    \hline \hline
    800 & 600 & 122 & 100 & 622 & 966
\end{tabular}
\end{center}

With $k = 1$ and for $p$ as $1$, $2$, and $\infty$ we get the following results:

\begin{center}
\begin{tabular}{c|c|c|c|c}
    $p$ & $s_p(c_1, \pi)$ & $s_p(c_2, \pi)$ & $s_p(c_3, \pi)$ & $\mathcal{F}_{1,p}(\pi)$\\
    \hline \hline
    $1$ & $1161$ & $1133$ & $916$ & $\{\{c_1\}\}$ \\
    $2$ & $1310.53$ & $1353.78$ & $1244.89$ & $\{\{c_2\}\}$ \\
    $\infty$ & $1522$ & $1666$ & $1710$ & $\{\{c_3\}\}$ \\
\end{tabular}
\end{center}
\end{example}

\section{Axiomatic Justifications}

There is a long tradition, going back to at least Arrow's Impossibility Theorem in \cite{Arrow1950-xw}, of looking for axioms that various voting systems satisfy or that characterize voting systems. Approval voting is no different. Indeed, there are various axiom sets that \textit{uniquely} characterize approval voting. In this section, we look at some of these axioms and discuss how to generalize them to uniquely characterize $k,p$-approval.

Throughout this section, and for the rest of the paper, we will assume both $k$ and $p$ are fixed, $\mathcal{C} = \{c_1, ..., c_n\}$, $\mathcal{O} = \{K \subset \mathcal{C} : |K| = k\}$, etc., as in definition \ref{def:p_aavs} above.

\subsection{Prior Results}

There are several sets of axioms which uniquely characterize approval voting with $k = 1$. We will first present some of these and the relevant characterizations before turning to generalizations of the axioms and new characterizations. 

\begin{remark}
We note that since profiles are just functions, they can be added and scaled. These are defined as $(\pi + \pi')(b) = \pi(b) + \pi'(b)$ and, for $\ell \in \R$, $(\ell \cdot \pi)(b) = \ell \cdot \pi(b)$. 
\end{remark}

Before defining the axioms, we create a preliminary definition. We need to be able to create a profile out of a single ballot.

\begin{definition}
Given $b \in \mathcal{I}$, we define the \textbf{single ballot profile} as $\pi^b:\mathcal{I} \to \R_{> 0}$ where $\pi^b(b') = \mathds{1}_{b' = b}$.
\end{definition}

This means that profiles are just linear combinations of the $\pi^b$ with non-negative coefficients.

We now define some axioms that are used to characterize approval voting. These axioms assume that $\mathcal{F}$ is a voting rule applied to profiles on approval ballots (such as those in definite \ref{def:p_aavs}). The axioms, which we find in \cite{Laslier2010-ki}, are: \begin{itemize}
    \item A voting rule, $\mathcal{F}$, satisfies \textbf{faithfulness} if when there is only one voter, their vote specifies the winners, i.e., for all $b \in \mathcal{I}$, $\mathcal{F}(\pi^b) = \{\{c\} : c \in b\}$.
    
    \item A voting rule, $\mathcal{F}$, satisfies \textbf{consistency} if the outcome of combining two profiles is the intersection of their outcomes (when its not empty), i.e., given two profiles, $\pi, \pi' \in \Pi$,

    $$\mathcal{F}(\pi) \cap \mathcal{F}(\pi') \neq \emptyset \Rightarrow \mathcal{F(\pi + \pi')} = \mathcal{F}(\pi) \cap \mathcal{F}(\pi').$$

    \item A voting rule, $\mathcal{F}$, satisfies \textbf{cancellation} if every candidate wins when they all have the same number of votes, i.e., $\mathcal{F}(\pi) = \mathcal{O}$ when for all $c, c' \in \mathcal{C}$, $s_{\infty}(c, \pi) = s_{\infty}(c', \pi)$.
    \item A voting rule, $\mathcal{F}$, satisfies \textbf{disjoint equality} if $\mathcal{F}(\pi^b + \pi^{b'}) = \{\{c\} : c \in b \cup b'\}$ when $b \cap b' = \emptyset$.
\end{itemize}

With these axioms, for $k = 1$, there are two well-known theorems.

\begin{thm}
\label{thm:axiomatic_approval_cancellation}
(\cite{Alos-Ferrer2006-hj}, \cite{Fishburn1978-pc}) Let $\mathcal{F}$ be a voting rule. $\mathcal{F} = \mathcal{F}_{1,\infty}$ if and only if it satisfies faithfulness, consistency, and cancellation. \end{thm}

This was proved in \cite{Fishburn1978-pc} but that proof also required an additional assumption called neutrality.\footnote{Neutrality refers to the idea that if the candidates were renamed, the outcome of the election would just be the outcome original outcome with the renaming applied.} Later, \cite{Alos-Ferrer2006-hj} provided a simpler proof without the assumption of neutrality.

\begin{thm}
\label{thm:axiomatic_approval_disjoint_equality}
(\cite{Brandl2022-vz}) Let $\mathcal{F}$ be a voting rule.  $\mathcal{F} = \mathcal{F}_{1,\infty}$ if and only if it satisfies faithfulness, consistency, and disjoint equality.
\end{thm}

\subsection{New Axioms}

When $k < 1$ and/or $p < \infty$, we can see that $k,p$-approval voting does not satisfy most of the axioms above. If $k > 1$ then $\mathcal{F}_{k,p}\left(\pi^b\right)$ consists of sets of size $k \neq 1$ and so it fails faithfulness. Additionally, even with $k = 1$, for $p < \infty$, $\mathcal{F}_{k,p}\left(\pi^{\{c_1\}} + \pi^{\{c_2, c_3\}}\right) = \{\{c_1\}\}$. If $n = 3$, this violates cancellation since all the candidates got the same number of votes. For any $n \geq 3$, this violates disjoint equality since 

$$\mathcal{F}_{k,p}\left(\pi^{\{c_1\}} + \pi^{\{c_2, c_3\}}\right) = \{\{c_1\}\} \neq \{\{c_1\}\} \cup \{\{c_2\}, \{c_3\}\} = \mathcal{F}_{k,p}\left(\pi^{\{c_1\}}\right) \cup  \mathcal{F}_{k,p}\left(\pi^{\{c_2, c_3\}}\right).$$ 

As such, we will need to generalize these axioms to characterize $k,p$-approval voting.

\begin{remark}
For the rest of this section, we will expand the definition of profiles so that $\Pi = \R_{\geq 0}^{\mathcal{I}}$, i.e., a profile, $\pi$, is a function from ballots to non-negative real numbers, not just non-negative integers.
\end{remark}

In order to generalize to $k > 1$, we introduce the winners function which takes a ballot as input and returns those committees with maximum intersection.

\begin{definition}
\label{def:winners_function}
We define the \textbf{winners function} $\mathcal{W}_k \colon \mathcal{I} \to \mathcal{O}$ as

$$\mathcal{W}_k(b) = \{K \in \mathcal{O} : |K \cap b| = min(|b|, k)\}.$$

For convenience, we let $\mathcal{W}_k(\emptyset) = \mathcal{O}$.

\end{definition}

In addition, to deal with $p < \infty$, we need to be able to transform profiles into functions that return scores, not votes.

\begin{definition}
For each $b \in \mathcal{I}$, we define the \textbf{unit score ballot profile} as

$$\theta_p^b = |b|^{1/p} \cdot \pi^b$$
\end{definition}

If we take an arbitrary profile, $\pi$, we can write it as

$$\pi = \sum_{b \in \mathcal{I}}{\beta_b \cdot \pi^b} \in span\left\{\pi^b : b \in \mathcal{I}\right\}$$

where the $\beta_b = \pi(b) \geq 0$. Similarly, for the same $\pi$ we can write it as

$$\pi = \sum_{b \in \mathcal{I}}{\gamma_b \cdot \theta_p^b}  \in span\left\{\theta_p^b : b \in \mathcal{I}\right\}$$

where $\gamma_b = \beta_b/|b|^{1/p} \geq 0$. The main difference between these two ways of writing $\pi$ is that, when written as an element of $span\left\{\theta_p^b : b \in \mathcal{I}\right\}$ we can calculate a candidates score as

$$s_p(c, \pi) = \sum_{b \in \mathcal{I}}{\gamma_b \cdot \mathds{1}_{c \in b}}$$

as opposed to how the score is calculated in definition \ref{def:p_aavs}.

Now that we have these definitions, we can generalize the axioms. 

\begin{definition}
\label{def:k_faithfulness}
A voting rule, $\mathcal{F}$, satisfies \textbf{$k$-faithfulness} if when there is only one voter, their vote specifies the winners, i.e., for all $b \in \mathcal{I}$, $\mathcal{F}(\pi^b) = \mathcal{W}_k(b)$.
\end{definition}

\begin{definition}
\label{def:p_cancellation}
A voting rule, $\mathcal{F}$, satisfies $p$-\textbf{cancellation} if every committee wins when every candidate has the same score, i.e., $\mathcal{F}(\pi) = \mathcal{O}$ when for all $c, c' \in \mathcal{C}$, $s_{p}(c, \pi) = s_{p}(c', \pi)$.
\end{definition}

\begin{definition}
\label{def:disjoint_equality}
A voting rule, $\mathcal{F}$, satisfies \textbf{$p$-disjoint equality} if $\mathcal{F}\left(\theta_p^b + \theta_p^{b'}\right) = \mathcal{W}_k(b \cup b')$ when $b \cap b' = \emptyset$.
\end{definition}

\begin{definition}
\label{def:positive_scaling}
A voting rule, $\mathcal{F}$, satisfies \textbf{positive scaling} if for $\ell \in \R_{> 0}$, $\mathcal{F}(\ell \pi) = \mathcal{F}(\pi)$ for all $\pi \in \Pi$.
\end{definition}

We can see that $k$-faithfulness generalizes faithfulness since $\mathcal{W}_1(b) = \{\{c\} : c \in b\}$. In addition, $p$-cancellation and $p$-disjoint equality generalize cancellation and disjoint equality (respectively) since $\theta_{\infty}^b = \pi^b$. Positive scaling, however, is an entirely new axiom. If $\mathcal{F}$ satisfies consistency, positive scaling is unnecessary when $\ell = x/y$ for $x, y \in \Z_{> 0}$ because

$$\mathcal{F}\left(\frac{x}{y}\pi\right)
= \mathcal{F}\left(\sum_{i = 1}^{y}{\frac{x}{y}\pi}\right)
= \mathcal{F}(x \pi) = \bigcap_{i = 1}^{x}{\mathcal{F}(\pi)}
= \mathcal{F}(\pi).$$

However, for $\ell \in \R_{> 0} \setminus \Q_{> 0}$, the above doesn't work since we can't write such an $\ell$ as the ratio of two positive integers. The table below briefly compares the original versus the generalized axioms.

\begin{center}
\begin{tabular}{c|c|c}
    Original Axiom & Generalized Axiom & Purpose of Generalization \\
    \hline \hline
    faithfulness & $k$-faithfulness & handle $k > 1$ \\
    consistency & N/A & N/A \\
    cancellation & $p$-cancellation & handle $p < \infty$ \\
    disjoint equality & $p$-disjoint equality & handle $p < \infty$ \\
    N/A & positive scaling & handle profiles with irrational images
\end{tabular}
\end{center}

In the following sub-sections, we turn to generalizing theorems \ref{thm:axiomatic_approval_cancellation} and \ref{thm:axiomatic_approval_disjoint_equality}.

\subsection{Faithfulness, Consistency, Positive Scaling, and Cancellation}

Here, we generalize theorem \ref{thm:axiomatic_approval_cancellation} in two ways: for $k > 1$ and for all $p \in [1, \infty]$.

\begin{thm}
\label{thm:axiomatic_p_approval_cancellation}
Let $\mathcal{F}$ be a voting rule. $\mathcal{F} = \mathcal{F}_{k,p}$ if and only if it satisfies $k$-faithfulness, consistency, $p$-cancellation, and positive scaling.
\end{thm}

First we prove a preliminary lemma.

\begin{lemma}
\label{lem:simple_ballot_equivilance}
Let $\mathcal{F}$ be a voting rule, $\mathcal{\pi} \in \Pi$ a profile, $b \in \mathcal{I}$ a ballot, $\ell > 0$ a positive real number. If $\mathcal{F}$ satisfies consistency, $p$-cancellation, and positive scaling, then
$$\mathcal{F}\left(\pi + \ell \cdot \theta_p^b \right) = \mathcal{F}\left(\pi + \sum_{c \in b}{\ell \cdot \theta_p^{\{c\}}} \right).$$
\end{lemma}

\begin{proof}
We denote the complement of $b$ as $a = \mathcal{I} \setminus b$. By $p$-cancellation and positive scaling, 

$$\mathcal{F}\left(\ell \cdot \theta_p^b + \ell \cdot \theta_p^a\right)= \mathcal{O}
= \mathcal{F}\left(\sum_{c \in b}{\ell \cdot \theta_p^{\{c\}}} + \ell \cdot \theta_p^a\right),$$

since, for each profile, all the candidates have a score of $\ell$. Since $\mathcal{O}$ is the set of all possible committees of size $k$, we use consistency to get

\begin{align*}
\mathcal{F}\left(\pi + \ell \cdot \theta_p^{b}\right)
& = \mathcal{F}\left( \pi + \ell \cdot \theta_p^{b} \right) \cap \mathcal{O} \\
& = \mathcal{F}\left(\pi + \ell \cdot \theta_p^{b} + \sum_{c \in b}{\ell \cdot \theta_p^{\{c\}}} + \ell \cdot \theta_p^{a} \right) \\
& = \mathcal{F}\left(\pi + \sum_{c \in b}{\ell \cdot \theta_p^{\{c\}}} + \ell \cdot \theta_p^{b} + \ell \cdot \theta_p^{a} \right) \\
& = \mathcal{F}\left(\pi + \sum_{c \in b}{\ell \cdot \theta_p^{\{c\}}}\right) \cap \mathcal{O} \\
& = \mathcal{F}\left(\pi + \sum_{c \in b}{\ell \cdot \theta_p^{\{c\}}}\right).
\end{align*}
\end{proof}

\begin{corollary}
\label{cor:simple_profile}
The outcome of $\mathcal{F}(\pi)$ only depends on the scores $s_p(c, \pi)$ for $c \in \mathcal{C}$.
\end{corollary}

\begin{proof}
Given a profile, $\pi$, we can create a new profile, $\pi'$, by starting with $\pi' \equiv 0$ and iteratively applying Lemma \ref{lem:simple_ballot_equivilance} to get $\pi'$ such that $\pi'(\{c\}) = s_p(c, \pi)$ and $\pi'(b) = 0$ for $|b| > 1$. In doing this, $\mathcal{F}(\pi) = \mathcal{F}(\pi')$.
\end{proof}

We now proceed to prove theorem \ref{thm:axiomatic_p_approval_cancellation}.

\begin{proof}
(Theorem \ref{thm:axiomatic_p_approval_cancellation}) ($\Rightarrow$) First, we show that $\mathcal{F}_{k,p}$ satisfies the axioms. $\mathcal{F}_{k,p}$ satisfies $k$-faithfulness since $s_p(c, \pi^b) = \mathds{1}_{c \in b}/|b|^{1/p}$ and so $\mathcal{F}_{k,p}(\pi^b) = \mathcal{W}_k(b)$.

For consistency, by basic properties of $\argmax$, if $\mathcal{F}_{k, p}(\pi) \cap \mathcal{F}_{k, p}(\pi') \neq \emptyset$, we get that

\begin{align*}
\mathcal{F}_{k, p}(\pi + \pi')
& = \argmax_{K \in\mathcal{O}}{\left\{ \sum_{c \in K}{s_p(c, \pi + \pi')} \right\}} \\
& = \argmax_{K \in\mathcal{O}}{\left\{ \sum_{c \in K}{s_p(c, \pi)} \right\}} \cap \argmax_{K \in\mathcal{O}}{\left\{ \sum_{c \in K}{s_p(c, \pi')} \right\}} \\
& = \mathcal{F}_{k, p}(\pi) \cap \mathcal{F}_{k, p}(\pi').
\end{align*}

For $p$-cancellation, if $s_p(c, \pi) = s_p(c', \pi)$ for all $c, c' \in \mathcal{C}$, we get that $\sum_{c \in K}{s_p(c, \pi)} = \sum_{c \in K'}{s_p(c', \pi)}$ for all $K, K' \in \mathcal{O}$ meaning that 

$$\mathcal{F}_{p}(\pi) = \argmax_{K \in\mathcal{O}}{\left\{ \sum_{c \in K}{s_p(c, \pi)} \right\}} = \mathcal{O}.$$

Finally, for positive scaling, let $\ell > 0$. So, $s_p(c, \ell \cdot \pi) = \ell \cdot s_p(c, \pi)$ and so we get

\begin{align*}
\mathcal{F}_{k, p}(\ell \cdot \pi)
& = \argmax_{K \in\mathcal{O}}{\left\{ \sum_{c \in K}{s_p(c, \ell \cdot \pi)} \right\}} \\
& = \argmax_{K \in\mathcal{O}}{\left\{ \sum_{c \in K}{\ell \cdot s_p(c, \pi)} \right\}} \\
& = \argmax_{K \in\mathcal{O}}{\left\{ \sum_{c \in K}{s_p(c, \pi)} \right\}} \\
& = \mathcal{F}_{k,p}(\pi).
\end{align*}

($\Leftarrow$) Now, assume $\mathcal{F}$ satisfies the axioms. Given a profile $\pi$, we construct a new profile, $\pi^*$.  To do this, we let $t_1 > t_2 > \cdots > t_z = 0$ be the sorted sequence of scores in $s_p(\mathcal{O}, \pi)$, i.e., all the scores actually obtained by some committee. (Note, we include $t_z = 0$ even if no committee has a score of zero). We also let $\kappa_i = \{c \in \mathcal{C} : s_p(c, \pi) = t_i\}$. So, we construct $\pi^*$ as

\begin{align*}
\pi^* = & (t_1 - t_2) \cdot \theta_p^{\kappa_1} \\
& + (t_2 - t_3) \cdot (\theta_p^{\kappa_1} \cup \theta_p^{\kappa_2}) \\
& \cdots \\
& + (t_{z - 1} - t_z) \cdot \left(\theta_p^{\kappa_1} \cup \cdots \cup \theta_p^{\kappa_{z - 1}}\right).
\end{align*}

First, we note that the scores in $\pi$ and $\pi^*$ are the same so, by corollary \ref{cor:simple_profile}, they have the same outcomes. Also, by $k$-faithfulness and positive scaling, the outcome of each row is

$$\mathcal{F}((t_{y} - t_{y + 1}) \cdot (\theta_p^{\kappa_1} \cup \cdots \cup \theta_p^{\kappa_{y}})) = \mathcal{W}_k(\kappa_1 \cup \cdots \cup \kappa_y).$$

Now, we note that for ballots $b_1, b_2 \in \mathcal{I}$, if $b_1 \subset b_2$ and $|b_2| \leq k$ then $\mathcal{W}(b_1) \supset \mathcal{W}(b_2)$ but if $b_1 \supset b_2$ and $|b_2| \geq k$ then $\mathcal{W}(b_1) \subset \mathcal{W}(b_2)$. 

So, there are three possibilities. If $| \kappa_1| > k$ then, by consistency, $\mathcal{F}(\pi) = \mathcal{F}(\pi^*) = \mathcal{W}_k(\kappa_1) = \mathcal{F}_{k,p}(\pi)$ as desired.  Similarly, if $|\kappa_1 \cup \cdots \cup \kappa_{z-1}| < k$, then $\mathcal{F}(\pi) = \mathcal{F}(\pi^*) = \mathcal{W}_k(\kappa_1 \cup \cdots \cup \kappa_{z-1}) = \mathcal{F}_{k,p}(\pi)$ as desired. Otherwise, there is a $y$ such that $|\kappa_1 \cup \cdots \cup \kappa_y| \leq k$ and $|\kappa_1 \cup \cdots \cup \kappa_y \cup \kappa_{y + 1}| > k$. We split the above profile into the two profiles

$$L = \sum_{i = 1}^{y}{(t_i - t_{i + 1}) \cdot (\theta_p^{\kappa_1} \cup \cdots \cup \theta_p^{\kappa_{i}}))}$$

and

$$R = \sum_{i = y + 1}^{z - 1}{(t_i - t_{i + 1}) \cdot (\theta_p^{\kappa_1} \cup \cdots \cup \theta_p^{\kappa_{i}}))}.$$

By consistency, we know that 

$$\mathcal{F}(L) = \mathcal{W}_k(\kappa_1 \cup \cdots \cup \kappa_y) = \left\{K \in \mathcal{O} : \kappa_1 \cup \cdots \cup \kappa_y \subset K \right\}$$

and

$$\mathcal{F}(R) = \mathcal{W}_k(\kappa_1 \cup \cdots \cup \kappa_y \cup \kappa_{y + 1}) = \left\{K \in \mathcal{O} : K \subset \kappa_1 \cup \cdots \cup \kappa_y \cup \kappa_{y + 1} \right\}.$$

So, we get that

\begin{align*}
\mathcal{F}(L) \cap \mathcal{F}(R) 
& = \left\{ K \in \mathcal{O} : \kappa_1 \cup \cdots \cup \kappa_y \subset K \subset \kappa_1 \cup \cdots \cup \kappa_y \cup \kappa_{y + 1} \right\} \\
& = \{\kappa_1 \cup \cdots \cup \kappa_y \cup S : S \subset \kappa_{y + 1} : |S| = k - |\kappa_1 \cup \cdots \cup \kappa_y|\} \\
& = \mathcal{F}_{k,p}(\pi),
\end{align*}

where the last line is true because: (a) the intersection is non-empty; and (b) this is all the committees with the candidates who ``clearly'' win and then any selection of the next tier candidates needed to get to $k$ candidates. In all cases, we get $\mathcal{F}(\pi) = \mathcal{F}_{k,p}(\pi)$ as desired. 
\end{proof}

This proof mirrors that in \cite{Alos-Ferrer2006-hj}. Lemma \ref{lem:simple_ballot_equivilance} and Corollary \ref{cor:simple_profile} mirror step 1 and step 2 in \cite{Alos-Ferrer2006-hj}. The construction of $\pi^*$ and the analysis of $\mathcal{F}(\pi^*)$ mirrors step 3 in \cite{Alos-Ferrer2006-hj}. The key changes are the introduction of $\mathcal{W}_k$ to allow for committees and $\theta_p^b$ to allow ballots of unit score instead of unit vote.

\subsection{Faithfulness, Consistency, Positive Scaling, and Disjoint Equality}

Here, we generalize Theorem \ref{thm:axiomatic_approval_disjoint_equality}. Unfortunately, unlike with Theorem \ref{thm:axiomatic_p_approval_cancellation}, although we are able to generalize to $p < \infty$, we are unable to generalize to $k > 1$.

\begin{thm}
\label{thm:axiomatic_p_approval_disjoint_equality}
Let $\mathcal{F}$ be a voting rule.  $\mathcal{F} = \mathcal{F}_{1,p}$ if and only if it satisfies faithfulness, consistency, $p$-disjoint equality, and positive scaling.
\end{thm}

\begin{proof}
($\Rightarrow$) We showed that $\mathcal{F}_{1,p}$ satisfies faithfulness, consistency, and positive scaling the proof of theorem \ref{thm:axiomatic_p_approval_cancellation}. It remains to show that $\mathcal{F}_{1,p}$ satisfies $p$-disjoint equality. Let $b, b' \in \mathcal{I}$ such that $b \cap b' = \emptyset$. So, for $c \in \mathcal{C}$, we get that

$$s_p(c, \theta_p^b + \theta_p^{b'}) = \mathds{1}_{c \in b} + \mathds{1}_{c \in b'}.$$

Since $b \cap b' = \emptyset$ at most one of the indicator functions above can return a $1$ and so we get

$$s_p(c, \theta_p^b + \theta_p^{b'}) = \mathds{1}_{c \in b \cup b'}.$$

Thus, we maximize the total score of the members of a committee by selecting as many candidates as possible from $b \cup b'$, i.e.,

$$\mathcal{F}_{1, p}(\theta_p^b + \theta_p^{b'}) = \argmax_{K \in\mathcal{O}}{\left\{ \sum_{c \in K}{s_p(c, \theta_p^b + \theta_p^{b'})} \right\}} = \mathcal{W}_k(b \cup b').$$

($\Leftarrow$) Let $\mathcal{F}$ be a voting rule that satisfies the axioms and $\pi \in \Pi$ an arbitrary profile where

$$\pi = \sum_{b \in \mathcal{I}}{\gamma_b \cdot \theta_p^b}$$

with $\{c\} \in \mathcal{F}(\pi)$ and $\{c'\} \in \mathcal{F}_{1,p}(\pi)$. We'll show that $\{c'\} \in \mathcal{F}(\pi)$ and $\{c\} \in \mathcal{F}_{1,p}(\pi)$. To do this, we first let

$$\pi[c'/c] = \sum_{b \in \mathcal{I}, c \notin b, c' \in b}{\gamma_b},
\bigspace{}\bigspace{}
\pi[c/c'] = \sum_{b \in \mathcal{I}, c' \notin b, c \in b}{\gamma_b}, \text{ and}
\bigspace{}\bigspace{}
\pi[\cdot/cc'] = \sum_{b \in \mathcal{I}, c', c \notin b}{\gamma_b}.$$

We create a new score profile, $\pi'$ via

$$\pi' =
\pi[c'/c] \cdot \theta_p^{\{c\}} + 
\pi[c/c'] \cdot \theta_p^{\{c'\}} + 
\pi[\cdot/cc'] \cdot \theta_p^{\{c, c'\}}.$$

So we can write $\pi + \pi'$ in two different ways. In the first way:

\begin{align*}
\pi + \pi' = &
\sum_{\substack{b \in \mathcal{I} \\ c \notin b, c' \in b}}{\pi(b)(\theta_p^b + \theta_p^{\{c\}})} +
 \sum_{\substack{b \in \mathcal{I} \\ c' \notin b, c \in b}}{\pi(b)(\theta_p^b + \theta_p^{\{c'\}})} + \\
& \sum_{\substack{b \in \mathcal{I} \\ c,c' \notin b}}{\pi(b)(\theta_p^b + \theta_p^{\{c, c'\}})} + 
\sum_{\substack{b \in \mathcal{I} \\ c,c' \in b}}{\pi(b)(\theta_p^b)}.
\end{align*}

We can see that $\{c\}, \{c'\} \in \mathcal{F}(\pi + \pi')$ by seeing they are returned by $\mathcal{F}$ for each term via: \begin{itemize}
    \item for each term in the first three summations, $\{c\}$ and $\{c'\}$ are returned due to positive scaling, $p$-disjoint equality, and faithfulness; and
    \item for each term in the last summation, $\{c\}$ and $\{c'\}$ are returned due to positive scaling, cancellation, and faithfulness.
\end{itemize}

Alternatively, we can write $\pi + \pi'$ as:

$$\pi + \pi' 
= \pi 
+ \pi[c/c'] \left(\theta_p^{\{c\}} + \theta_p^{\{c'\}}\right)
+ \left(\pi[c'/c] - \pi[c/c']\right)\left(\theta_p^{\{c\}}\right)
+ \pi[\cdot/cc']\left(\theta_p^{\{c, c'\}}\right).$$

Note that since $\{c'\} \in \mathcal{F}_{1,p}(\pi)$, $\pi[c'/c] \geq \pi[c/c']$. First, we show that $\{c\}$ is returned for each term by: \begin{itemize}
    \item $\{c\} \in \mathcal{F}(\pi)$ by assumption;
    \item $\{c\} \in \mathcal{F}\left(\pi[c/c'] \left(\theta_p^{\{c\}} + \theta_p^{\{c'\}}\right)\right)$ by disjoint equality;
    \item $\{c\} \in \mathcal{F}\left(\left(\pi[c'/c] - \pi[c/c']\right)\left(\theta_p^{\{c\}}\right)\right)$ by faithfulness; and
    \item $\{c\} \in \mathcal{F}\left(\pi[\cdot/cc']\left(\theta_p^{\{c, c'\}}\right)\right)$ by faithfulness.
\end{itemize}

So, since $\{c\}$ is returned by $\mathcal{F}$ for each term, we know that $\mathcal{F}(\pi + \pi')$ is the intersection of what $\mathcal{F}$ returns for each term and that $\{c\} \in \mathcal{F}(\pi + \pi')$. Given this, if $\{c'\} \notin \mathcal{F}(\pi)$ then $\{c'\} \notin \mathcal{F}(\pi + \pi')$ which contradicts what we proved above. Conversely, if $\{c\} \notin \mathcal{F}_{1,p}(\pi)$, then $\pi[c'/c] - \pi[c/c'] > 0$ which means that $c' \notin \mathcal{F}(\pi + \pi')$ which also contradicts what we proved above. Thus, $\mathcal{F}(\pi) = \mathcal{F}_{1,p}(\pi)$.
\end{proof}

The proof above is almost the same as that in \cite{Brandl2022-vz} unit score ballots instead of (regular) ballots.

\section{Utility Maximization and Maximum Likelihood Estimation Justifications}

In this section, we look at justifications of $k,p$-approval voting in two different, but related, ways. The first is relates to utility maximization, i.e., investigating when a voting system maximizes the total utility of the voters. The second relates to likelihood maximization. Here, we assume there is a ``correct'' answer (i.e., ``correct'' committee) and that the votes are ``noisy'' signals that can be aggregated to help find the correct committee. A voting system is a maximum likelihood estimator if it finds the ``best answer'' to a given noise model.

\subsection{Preliminaries}
Before we discuss this, we present some definitions specific to this topic.

\begin{definition}
\label{def:utlity}
A \textbf{utility function} is a function $U:\mathcal{I} \times \mathcal{O} \to \R$. The interpretation is once we know a voter's ballot, we know their utility for each possible outcome. A voting rule, $\mathcal{F}$, is \textbf{utility maximizing} (UM) with respect to $U$ if, for all $\pi \in \Pi$,

$$\mathcal{F}(\pi) = \argmax_{K \in \mathcal{O}}{\left\{\sum_{b \in \mathcal{I}}{\pi(b) \cdot U(b, K)}\right\}}.$$
\end{definition}

\begin{definition}
\label{def:mle}
A \textbf{noise model} is a function $P:\mathcal{O} \times \mathcal{I} \to \R_{\geq 0}$ such that for any fixed $K \in \mathcal{O}$, 

$$\sum_{b \in \mathcal{I}}{P(K, b)} = 1.$$

That is, when we hold $K \in \mathcal{O}$ constant, $P$ forms a probability distribution over $\mathcal{I}$. A voting rule is a \textbf{maximum likelihood estimator} (MLE) with respect to $P$ if, for all $\pi \in \Pi$,

$$\mathcal{F}(\pi) = \argmax_{K \in \mathcal{O}}{\left\{ 
\prod_{b \in \mathcal{I}}{P(K, b)^{\pi(b)}} \right\}}.$$
\end{definition}

A voting rule as UM is a simple concept. The idea of a voting rule as an MLE is explained in detail in \cite{Conitzer2005CommonVR}. Briefly, the idea is that there is a ``correct'' outcome, $K^*$, and each voter is trying to ``guess'' the correct outcome in a noisy environment. In this context, a voting rule is a method for taking the noisy guesses and trying to determine the correct outcome.

\begin{example}
\label{ex:plurality_um_mle}
We look back at plurality voting as in example \ref{ex:plurality_voting}. If it is the case that each voter gets a utility of $1$ if their chosen candidate wins and $0$ otherwise, then plurality voting returns the candidate that maximizes the total utility.

Similarly, assume if the ``correct'' candidate is some $c^*$ then each voter has some $p_1$ chance of voting for that candidate and $p_2 < p_1$ chance of voting for any other candidate. In this case, plurality voting returns the maximum likelihood estimator for the correct candidate.
\end{example}

We can link the two concepts via the following theorems.

\begin{thm}
\label{thm:um_mle}
(\cite{Pivato2013-we}) A voting rule is UM for some $U$ if and only if it is MLE for some $P$.
\end{thm}

The proof is in \cite{Pivato2013-we} but the key transformation between $U$ and $P$ is

$$U(b, K) \sim \ln{P(K, b)}.$$

This allows us to go from maximizing a sum to maximizing a product and vice versa.

\subsection{k,p-Approval Voting as UM and MLE}

We now look at the specific forms of $U$ and $P$ for which $k,p$-approval voting is UM and MLE.

\begin{thm}
\label{thm:norm_approval_voting_general_u}
$\mathcal{F}_{k, p}$ is UM for

$$U(b, K) = U_0 + \frac{\alpha |b \cap K| + \beta |b \setminus K|}{| b |^{1/p}}$$

with $U_0 \in \R$ and $\alpha > \beta$.
\end{thm}

\begin{proof}
Let $m = \sum_{b \in \mathcal{I}}{\pi(b)}$. For any $K \in \mathcal{O}$,
\begin{align*}
\sum_{b \in \mathcal{I}}{\pi(b) \cdot U(b, K)}
& = \sum_{b \in \mathcal{I}}{U_0 + \frac{\alpha |b \cap K| + \beta |b \setminus K|}{| b |^{1/p}}} \\
& = mU_0 + \alpha \sum_{c \in K}{\sum_{b \in \mathcal{I}}{\frac{\pi(b)}{| b |^{1/p}} \cdot \mathds{1}_{c \in b}}} + \beta \sum_{c \notin K}{\sum_{b \in \mathcal{I}}{\frac{\pi(b)}{| b |^{1/p}} \cdot \mathds{1}_{c \in b}}} \\
& =  mU_0 + \alpha \sum_{c \in K}{s(c)} + \beta \sum_{c \notin K}{s(c)}.
\end{align*}

Since $\alpha > \beta$ this is maximized by having $K$ be the candidates with the highest scores.
\end{proof}

For satisfaction approval voting, or $k,1$-approval voting in general, when $\beta = 0$, this utility function is used in \cite{Brams2010-cr} to justify satisfaction approval voting. 

One might wonder if we can generalize this further by adding terms for the candidates we did not consider, i.e., $|K \setminus b|$ and $|\mathcal{C} \setminus (K \cup b)|$. To try this, for some fixed ballot $b$, we let $w = |b \cap K|$ and $\ell = |b \setminus K|$. Writing out a (seemingly) more general form gets us

\begin{align*}
U(b, K) & = U_0 + \frac{\alpha w + \beta \ell + \gamma (k - w) + \delta (m - w - \ell - (k - w))}{| b |^{1/p}} \\
& = U_0 + \frac{\gamma k + \delta (m - k)}{| b |^{1/p}} + \frac{(\alpha - \gamma) w + (\beta - \delta) \ell)}{| b |^{1/p}}.
\end{align*}

Since $k$ and $m$ are constants, this is of the same form as the equation of above with

$$U'_0 = U_0 + \frac{\gamma k + \delta (m - k)}{| b |^{1/p}}, 
\bigspace{}\bigspace{}
\alpha' = \alpha - \gamma, \text{ and } 
\bigspace{}\bigspace{}
\beta' = \beta - \delta.$$

As such, nothing is gained by adding in these additional terms. That is to say, we do not gain any additional mathematical expressive power. However, adding these terms could be useful when trying to determine utility functions in a practical context.

Theorems \ref{thm:um_mle} and \ref{thm:norm_approval_voting_general_u} also gives us the following corollary.

\begin{corollary}
\label{cor:norm_approval_voting_mle}
$\mathcal{F}_{k, p}$ is MLE for

$$P(K, b) \sim \alpha^{| b \cap K | / | b |^{1/p}} \cdot \beta^{| b \setminus K |/| b |^{1/p}}$$

when $\alpha > \beta$.
\end{corollary}

For $p = \infty$ the formula in corollary \ref{cor:norm_approval_voting_mle} has an interesting interpretation. $P(K, b)$ reduces to

$$\P(K, b) \sim \alpha^{| b \cap K |} \cdot \beta^{| b \setminus K |}.$$

Given a ``correct'' committee $K$ we posit a voter who looks at the candidates one at a time and either approves of each of them or not independently. If she approves of a candidate $c \in K$ with probability $p_1$ and a candidate $c \in K^c = \mathcal{C} \setminus K$ with probability $p_2$ then we get

\begin{align*}
P(K, b) & = \left( \prod_{c \in K \cap b}{p_1} \right) \cdot \left( \prod_{c \in K \setminus b}{1 - p_1} \right) \cdot \left( \prod_{c \in K^c \cap b}{p_2} \right) \cdot \left( \prod_{c \in K^c \setminus b}{1 - p_2} \right)\\
& = p_1^{| K \cap b |} \cdot (1 - p_1)^{| K \setminus b |} \cdot p_2^{| K^c \cap b |} \cdot (1 - p_2)^{| K^c \setminus b |} \\
& = p_1^w \cdot (1 - p_1)^{(k - w)} \cdot p_2^{\ell} \cdot (1 - p_2)^{(n - k - \ell)} \\
& = (1 - p_1)^k \cdot \left(\frac{p_1}{1 - p_1}\right)^w \cdot (1 - p_2)^{(n - k)} \cdot \left( \frac{p_2}{1 - p_2} \right)^{\ell}.
\end{align*}

Because $(1 - p_1)^k$ and $(1 - p_2)^{(n - k)}$ are constants we get

$$P(K, b) \sim \left(\frac{p_1}{1 - p_1}\right)^w \left( \frac{p_2}{1 - p_2} \right)^{\ell}.$$

When $p_1 > p_2$, this is in the same form as the expression in corollary \ref{cor:norm_approval_voting_mle}. Thus, at least for approval voting, this has a reasonable interpretation in terms of a voter selecting candidates one-by-one. For other norms, however, no such interpretation seems readily available.

\section{Distance Rationalizability Justifications}

Finally, we turn to distance rationalizability. Distance rationalizability is a concept used to justify various voting rules including approval voting \cite{Elkind2009-ot}. In general, a voting rule is distance rationalizable if there is a class of ``consensus elections'' (where the voters all agree on the correct outcome) and the voting rule is equivalent to finding the nearest consensus election and returning the consensus of that election.

\subsection{Preliminaries}

Until now, we've considered profiles which can be thought of as counting the number of voters who submitted each ballot. However, an election can also be thought of as a ``list'' showing which ballot each voter submitted. Based on this, we define an election as follows.

\begin{definition}
An \textbf{election} is a finite \textit{ordered} sequence of ballots, i.e., something of the form

$$E = \langle v_1, ..., v_m \rangle$$

where each $v_i \in \mathcal{I}$.
\end{definition}

\begin{remark}
\label{rem:fixed_m}
Throughout this section, we will assume all elections have some fixed number, $m$, of votes.
\end{remark}

Note that every election induces a profile such that $\pi(b) = |\{i \in \{1, ..., m\} : v_i = b\}|$.\footnote{This also means that the considerations in the previous section could be written in terms of elections. If one wished to do so, one would sum over the votes in the election instead of the ballots in $\mathcal{I}$.} As such, it makes sense to apply $\mathcal{F}$ to elections. Thus, we will, with reckless abandon, abuse notation and write $\mathcal{F}(E)$. Before jumping into the details, we give an example.

\begin{example}
\label{ex:plurality_dist_rat}
We return to plurality voting in example \ref{ex:plurality_voting}. We let $\mathcal{S}$ be the set of ``unanimous'' elections, i.e., the $n$ elections, $E_c$, where each voter votes for candidate $c$. Given two elections, $E = \langle v_1, ..., v_m \rangle$ and $E' = \langle v'_1, ..., v'_m \rangle$, let the distance between them be

$$d(E, E') = |\{i : v_i \neq v'_i\}|.$$

With this metric, we can see that for a given $c \in \mathcal{C}$,

$$d(E, E_c) = |\{i : v_i \neq c\}| = m - s(c).$$

So the closet unanimous elections are the $\{E_c : c \in \mathcal{F}(E)\}$. In other words, $\mathcal{F}$ returns exactly those candidates whose unanimous election is closest. As we will see with the definitions below, this makes plurality voting distance rationalizable.
\end{example}

We now lay out some definition from \cite{Elkind2009-ot} that we will need for our discussion.

\begin{definition}
\label{def:unanimous_elections}
For $K \in \mathcal{O}$ we define

$$\mathcal{U}_K = \left\{ E : K \subset \bigcap_{i = 1}^{m}{v_i} \right\}$$

and

$$\mathcal{U} = \bigcup_{K \in \mathcal{O}}{\mathcal{U}_K}.$$

The set $\mathcal{U}$ is the set of \textbf{unanimous elections}.
\end{definition}

There are other consensus classes in the literature with respect to approval ballots (\cite{Elkind2009-ot}) as well as relevant consensus classes and for other types of ballots (\cite{Elkind2009-ot}, \cite{Nitzan1981-ni}). Here, however, we focus on $\mathcal{U}$ as defined in definition \ref{def:unanimous_elections}.

\begin{definition}
\label{def:cand_dist}
Given a metric, $d$, on the space of elections, the \textbf{$(\mathcal{U}, d)$-score} of a committee $K \in \mathcal{O}$ is the distance from $E$ to the nearest unanimous election where $K$ is a winner, i.e.,

$$d_{\mathcal{U}}(K; E) = \min_{E' \in \mathcal{U}_{K}}{d(E, E')}.$$

The set of \textbf{$(\mathcal{U}, d)$-winners} are those committees with minimum scores.
\end{definition}

There are only a finite number of elections with $m$ votes so the minimum always exists.

\begin{definition}
\label{def:dist_rat}
Let $d$ be a metric between elections. A voting rule, $\mathcal{F}$, is \textbf{$(\mathcal{U}, d)$-distance rationalizable} if for all elections, the winners $\mathcal{F}$ returns is exactly the set of $(\mathcal{U}, d)$-winners.
\end{definition}

At least in cases where a voting rule always returns the unanimous decision for unanimous elections, i.e., where $\mathcal{F}\left( \mathcal{U}_K \right) = K$, one could use the Jaccard distance (\cite{Levandowsky1971-dz}) on the outcomes, i.e.,

$$d(E, E') = 1 - \frac{|\mathcal{F}(E) \cap \mathcal{F}(E')|}{|\mathcal{F}(E) \cup \mathcal{F}(E')|}.$$

Under this metric, almost any voting rule can be made distance rationalizable with respect to the unanimity class. However, this somewhat trivializes the concept of distance rationalizability since it directly measures how different the outcomes, not the elections, are. More interesting is when a we have a metric, $d$, on $\mathcal{I}$ and we use this to induce a metric, $\hat{d}$, on elections via

$$\hat{d}(E, E') = \sum_{i = 1}^{m}{d(v_i, v'_i)}.$$

\begin{definition}
\label{def:vote_dist_rat}
A voting rule is \textbf{vote distance rationalizable} if there is a metric, $d$, on $\mathcal{I}$ such that the rule is distance rationalizable with respect to the induced metric $\hat{d}$.
\end{definition}

If such a metric can be found, one can look at how ``reasonable'' or ``natural'' it is as a way of justifying a voting rule. In the remainder of this section, we look into the existence of such metrics.

\subsection{Distance Rationalizability  of Approval Voting}

It was shown in \cite{Elkind2009-ot} that, with $k = 1$, approval voting (i.e. $1,\infty$-approval voting) is vote distance rationalizable. Here, we generalize that to $k,\infty$-approval voting for any $k \geq 1$.

\begin{thm}
\label{thm:approval_voting_dist_rat}
For any $k \geq 1$, $k,\infty$-approval voting is vote distance rationalizable with respect to $\mathcal{U}$ with the metric $d(v, v') = |v \triangle v'|$ where $\triangle$ represents the symmetric difference between two sets.
\end{thm}

\begin{proof}
For $E' = \langle v'_1, ..., v'_m \rangle \in \mathcal{U}_K$,
\begin{align*}
\hat{d}(E, E') & = \sum_{i = 1}^{m}{d(v_i, v'_i)} \\
& = \sum_{i = 1}^{m}{|v_i \triangle v'_i|} \\
& = \left(\sum_{i = 1}^{m}{\left|\left(v_i \cap K\right) \triangle \left(v'_i \cap K\right)\right|}\right) + \left(\sum_{i = 1}^{m}{\left|\left(v_i \cap K^C\right) \triangle \left(v'_i \cap K^C\right)\right|}\right) \\
& = \left(\sum_{c \in K}{\sum_{i = 1}^{m}{\left(1 - \mathds{1}_{c \in v_i}\right)}}\right) + \left(\sum_{i = 1}^{m}{\left|\left(v_i \cap K^C\right) \triangle \left(v'_i \cap K^C\right)\right|}\right) \\
& = \left(\sum_{c \in K}{m - s_\infty(c)}\right) + \left(\sum_{i = 1}^{m}{\left|\left(v_i \cap K^C\right) \triangle \left(v'_i \cap K^C\right)\right|}\right) \\
& = mk - \sum_{c \in K}{s_\infty(c)}  + \left(\sum_{i = 1}^{m}{\left|\left(v_i \cap K^C\right) \triangle \left(v'_i \cap K^C\right)\right|}\right).
\end{align*}

For fixed $K$, this is minimized by picking the consensus election $E' = \langle v'_1, ..., v'_m \rangle  \in \mathcal{U}_K$ such that for all $i \in \{1, ..., m\}$, $v'_i = K \cup (v_i \setminus K)$. Doing this, we get

$$\hat{d}(E, E') = mk - \sum_{c \in K}{s_{\infty}(c)}.$$

This is minimized by picking $K$ with the highest scores, i.e., $K \in \mathcal{F}_{k, \infty}(E)$ as desired.
\end{proof}

\subsection{Distance Rationalizability When k = 1}

We now turn to general $1,p$-approval voting for all $p \in [1, \infty]$. To do this, we construct a graph, $G(\mathcal{I}, Z)$, where the nodes are the ballots and the edges are

$$Z = \{\{b, b'\} \in \mathcal{I}^2 : |b \triangle b'| = 1\}.$$

We also have a cost function on the edges,

$$c(\{b, b \cup \{c\}\}) = \frac{1}{|b|^{1/p}}.$$

With this, the distance, $d$, between two ballots is the sum of the costs on the minimum cost path between the two ballots. We can think of the graph as $n$ layers, each corresponding to votes of the same cardinality and layers corresponding to larger cardinalities as ``above'' those with lower cardinalities. We can see an example with $n = 4$ in Figure \ref{fig:layer_graph}.

\begin{figure}[htbp]
    \centering
    \includegraphics[width=\textwidth]{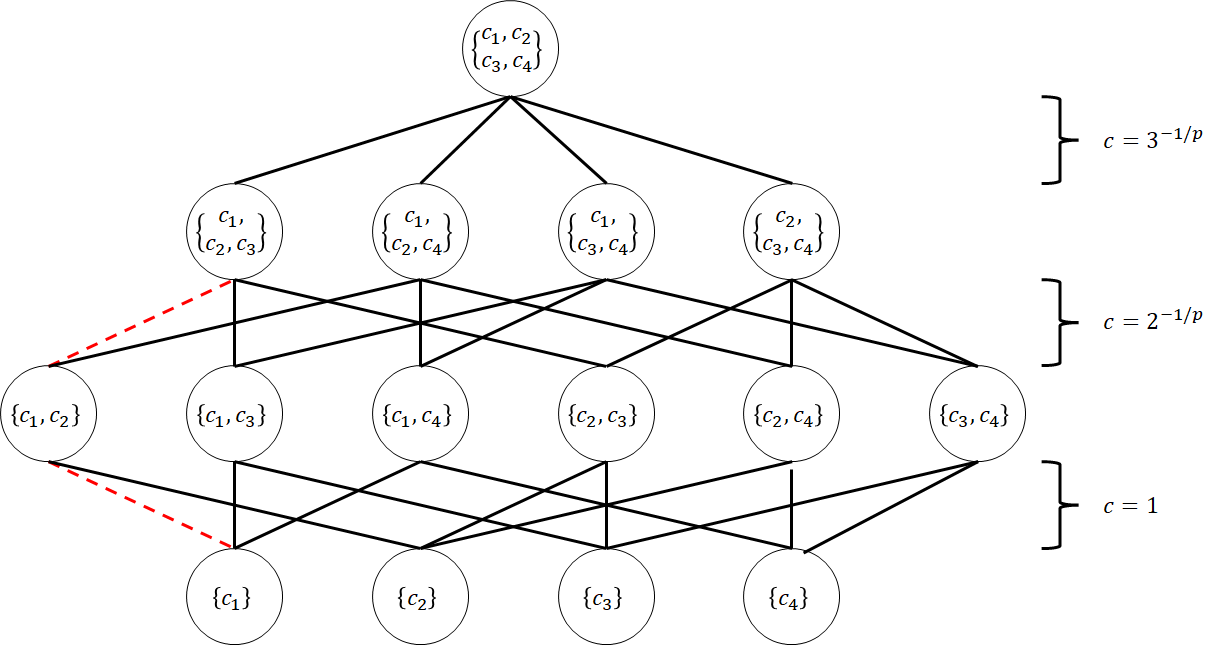}
    \caption{a ``layer'' graph with $n = 4$. Following the path indicated by the two red dashed lines, we can see $d(\{c_1\}, \{c_1, c_2, c_3\}) = 1 + 2^{-1/p}$.}
    \label{fig:layer_graph}
\end{figure}

\begin{thm}
\label{thm:dist_rat_k_1}
For $p \in [1, \infty]$, any $1,p$-approval voting rule is vote distance rationalizable with respect to $\mathcal{U}$ and $d$ as defined above.
\end{thm}

\begin{proof}
Fix an election $E$ and pick a candidate $c$. For each $v \in E$, $v' = v \cup \{c\}$ is the closest ballot containing $c$. If $c \in v$ then $v' = v$ so this is obvious. Otherwise, traversing from $v$ to $v'$ on $G$ requires just traversing up just one layer with a distance of $1/| v |^{1/p}$. For any other $v'' \ni c$, traversing from $v$ to $v''$ will require either traversing either down a layer (at a cost greater than $1/| v |^{1/p}$ since $z^p$ is increasing in $z$) or going up more than one layer (which includes a cost of $1/| v |^{1/p}$ and then more).

So, to find the nearest election $E' \in {\mathcal{U}}_{\{c\}}$, we need only add $c$ to each vote in $E$. The distance is

$$d(E, E')
= \sum_{i = 1}^{m}{\frac{1}{| v_i |^{1/p}} \mathds{1}_{c \notin v_i}}
= \sum_{i = 1}^{m}{\frac{1}{| v_i |^{1/p}} - \frac{1}{| v_i |^{1/p}} \mathds{1}_{c \in v_i}}
= \sum_{i = 1}^{m}{\frac{1}{| v_i |^{1/p}}} - s(c).$$

This is minimized by taking $c$ to be the candidate (or one of the candidates) with the highest score.
\end{proof}

Of course, one can question if the metric above is ``natural'' enough. As a rough justification, one can see this as a ``cost'' to adding a candidate to a vote. Given a vote, the cost of adding a candidate decreases as the number of candidates already in the vote increases. This makes sense if we consider that adding a candidate to a ``big'' vote is not changing it as much as adding a candidate to a ``small'' vote.

\subsection{Distance Rationalizability  for General k and p}

We now turn to the question of distance rationalizability for $k,p$-approval voting for $k > 1$ and $p < \infty$. However, as we will show, for $p < \infty$ and $k > 1$, $k,p$-approval voting is not vote distance rationalizable. Before we can do this, however, we need a lemma.

\begin{lemma}
\label{lem:delta_triangle}
Let $d$ be a metric on ballots, $E = \langle v_1, ..., v_m \rangle$ an election, and $U = \langle u_1, ..., u_n \rangle$ one of the nearest unanimous elections (with respect to $\hat{d}$). If we let $\tilde{E}$ be $E$ with one vote changed to match $U$ i.e. for some $\ell \in \{1, ..., m\}$, $\tilde{E} = \{v_1, ..., v_{\ell - 1}, u_l, v_{\ell + 1}, ..., v_m\}$ then $U$ is still among the closest unanimous elections to $\tilde{E}$.
\end{lemma}

\begin{proof}
Let $U' = \langle u'_1, ..., u'_m \rangle$ be another unanimous election. We show that $\hat{d}(\tilde{E}, U) \leq \hat{d}(\tilde{E}, U')$. We note that

$$\hat{d}(\tilde{E}, U) = \sum_{i = 1}^{\ell - 1}{d(v_1, u_i)} + \sum_{i = \ell + 1}^{m}{d(v_i, u_i)} = \hat{d}(E, U) - d(v_\ell, u_\ell),$$

and

$$\hat{d}(\tilde{E}, U') = \sum_{i = 1}^{\ell - 1}{d(v_1, u'_i)} + d(u_\ell, u'_\ell) + \sum_{i = \ell + 1}^{m}{d(v_i, u_i)} = \hat{d}(E, U) - d(v_\ell, u'_\ell) + d(u_\ell, u'_\ell)$$

Since $d$ is a metric it satisfies the triangle inequality we get

$$d(u_\ell, u'_\ell) + d(v_\ell, u_\ell) \geq d(v_\ell, u'_\ell) 
\Rightarrow - d(v_\ell, u_\ell) \leq - d(v_\ell, u'_\ell) + d(u_\ell, u'_\ell)$$

This means that

$$\hat{d}(\tilde{E}, U) = \hat{d}(E, U) - d(v_\ell, u_\ell) \leq \hat{d}(E, U) - d(v_\ell, u'_\ell) + d(u_\ell, u'_\ell) = \hat{d}(\tilde{E}, U').$$
\end{proof}

Now, we can prove our negative result.

\begin{thm}
\label{thm:norm_av_not_dr}
For $p \neq \infty$ and $k > 1$, $k,p$-approval voting is not vote distance rationalizable.
\end{thm}

\begin{proof}
We assume otherwise for the sake of a contradiction. Since there are at least $k + 1 \geq 3$ candidates, we name them

$$\mathcal{C} = \{a_1, a_2, a_3, b_1, ..., b_{k - 2}, c_1, ..., c_{n - k - 1}\}.$$

We setup an election, $E$, with the following votes:

\begin{itemize}
    \item $m_1$ votes for $\{a_1, a_3, b_1, ..., b_{k - 2}\}$;
    \item $m_2$ votes for $\{a_1, b_1, ..., b_{k - 2}\}$; and 
    \item $m_3$ votes for $\{a_2, b_1, ..., b_{k - 2}\}$
\end{itemize}

where we set $m_1, m_2, m_3 \in \Z_{\geq 0}$ such that $m_2$ is big enough that $a_1$ beats both $a_2$ and $a_3$ and

$$\frac{m_3}{k^{1/p}} < \frac{m_1}{k^{1/p}} < \frac{m_3}{(k - 1)^{1/p}}.$$

This is always possible since $z^p$ is strictly increasing in $z$. This means that in $E$, $a_1$ beats $a_2$ and $a_2$ beats $a_3$ and so $\mathcal{F}_{k,p}(E) = \{\{a_1, a_2, b_1, ..., b_{k - 2}\}\}$. Letting $U$ be the nearest election in $\mathcal{U}$ we know that its consensus is $\mathcal{F}_{k,p}(E)$.

Let $\tilde{E}$ be the election resulting from changing the last $m_3$ votes to match $U$. This means that these votes each have at least $k$ candidates which drops $a_2$'s score below $a_3$'s in $\tilde{E}$. By Lemma \ref{lem:delta_triangle}, the closest election should still be $U$. But, in $\tilde{E}$, $a_3$ beats $a_2$ so $U$ can't be the closest consensus election. This gives a contradiction.
\end{proof}

More succinctly, because adding a candidate to a vote decreases the score of all the candidates in that vote, this can change the outcomes for all the other candidates.

\section{Conclusion}

In this paper, we have looked at generalizations of approval voting where an approval ballot is viewed as analogous to a $0,1$-vector normalized by different $p$-norms. In many cases, any norm could, in fact, be used in place of a $p$-norm.

The idea that one cannot just add a candidate to an approval ``for free'' is a natural one - a person who approves of one candidate can be considered to want that candidate more than a person who votes for ten candidates wants any of those candidates. Alternatively, in a machine learning context, an agent that reports one possibility might be considered to have more confidence in that possibility that an agent that returns ten possibilities. The degree to which this desire, or confidence, decreases can be calibrated by picking the correct $p$-norm (or other norm). As such, it is helpful to see that there are several justifications for most of these methods.

The main exception to this in the paper is with respect to distance rationalizability for $k,p$-approval voting for $k > 1$ and $p < \infty$. The degree to which this is a problem depends on the context where a voting method is being used.

Looking into $k,p$-approval voting is a new area. Concurrent work in this area is looking into, among other things, the degree to which the outcome can change as $p$ changes, the Condorcet efficiency of $k,p$-approval voting for different $p$, the montonicity of $k,p$-approval voting as defined here as well as in an instant runoff context, variations where the vote is not a $0,1$-vector but rather an $a,b$-vector for some $0 < a < b$, etc.

\printbibliography

\end{document}